\documentclass[a4paper,12pt]{article}
\pdfoutput=1

\usepackage{test,caption}
\usepackage[colorlinks,citecolor=blue,urlcolor=magenta]{hyperref}

\usepackage[style=ext-authoryear-comp,%numeric-comp,
sorting=nyt,
dashed=false, % don't use dash for repeated author
maxcitenames=2, % maximum names to cite in text
maxbibnames=99, % maximum names to list in reference
uniquelist=false,
uniquename=false,
giveninits=true, % first name initials
natbib, % use natbib commands
date=year % year only
]{biblatex}

\AtBeginRefsection{\GenRefcontextData{sorting=ynt}}
\AtEveryCite{\localrefcontext[sorting=ynt]}

\DeclareFieldFormat{pages}{#1} % suppress pp.
\renewbibmacro{in:}{\ifentrytype{article}{}{\printtext{\bibstring{in}\intitlepunct}}} % delete In:

\DeclareFieldFormat[article,inbook,incollection,inproceedings,patent,thesis,unpublished]{titlecase:title}{\MakeSentenceCase*{#1}} % change title to sentence case

\usepackage[inline,shortlabels]{enumitem}
\setlist[enumerate,1]{label=(\roman*)}

\usepackage[margin = 1.25in]{geometry}
\usepackage{setspace}
\title{Unbalanced Growth and Land Overvaluation}

\author{Tomohiro Hirano\thanks{Department of Economics, Royal Holloway, University of London and The Canon Institute for Global Studies. Email: \href{mailto:tomohiro.hirano@rhul.ac.uk}{tomohiro.hirano@rhul.ac.uk}.} \and Alexis Akira Toda\thanks{Department of Economics, Emory University. Email: \href{mailto:alexis.akira.toda@emory.edu}{alexis.akira.toda@emory.edu}.}}

\numberwithin{equation}{section}
\numberwithin{lem}{section}

\begin{document}

\maketitle

\begin{abstract}

Historical trends suggest the decline in importance of land as a production factor but its continued importance as a store of value. Using an overlapping generations model with land and aggregate uncertainty, we theoretically study the long-run behavior of land prices and identify economic conditions under which land becomes overvalued on the long-run trend relative to the fundamentals defined by the present value of land rents. Unbalanced growth together with the elasticity of substitution between production factors plays a critical role. Around the trend, land prices exhibit recurrent stochastic fluctuations, with expansions and contractions in the size of land overvaluation.

\medskip


\medskip

\end{abstract}

\section{Introduction}

As economies develop and per capita incomes rise, the importance of land as a factor of production diminishes. Figure \ref{fig:US} plots the employment share of agriculture in the United States against time (left panel) and real per capita GDP (right panel). These graphs show that the employment share of agriculture has significantly declined over the past three centuries along economic development.\footnote{See Appendix \ref{subsec:data_US} for data description. Figure \ref{fig:US} updates \citet[Figure 2]{KongsamutRebeloXie2001}.} The trend is similar if we define the ``land-intensive'' sector by combining agriculture, mining, and construction. This trend is not specific to the United States: Figure \ref{fig:country} plots the GDP share of agriculture (left panel) and the employment share of agriculture (right panel) against per capita GDP for 173 countries in 2023.\footnote{See Appendix \ref{subsec:data_country} for data description. Figure \ref{fig:country} updates \citet[Figure 2]{Echevarria1997}.}

\begin{figure}[!htb]
    \centering
    \includegraphics[width=0.48\linewidth]{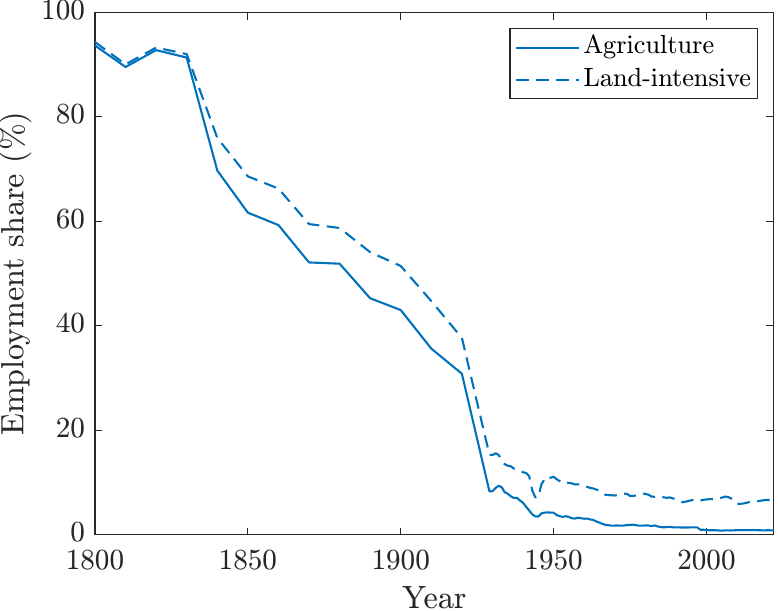}
    \includegraphics[width=0.48\linewidth]{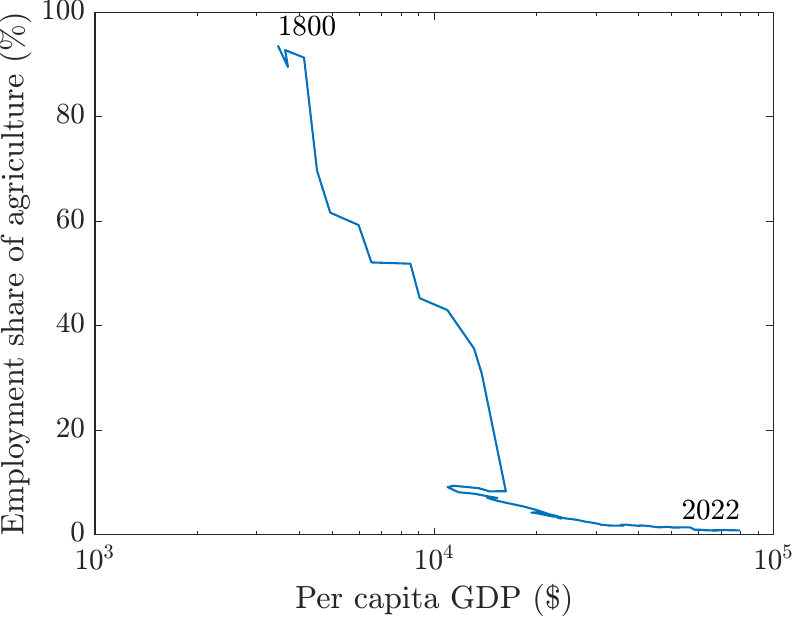}
    \caption{Employment share of land-intensive sectors in the United States.}
    \label{fig:US}
\end{figure}
\begin{figure}[!htb]
    \centering
    \includegraphics[width=0.48\linewidth]{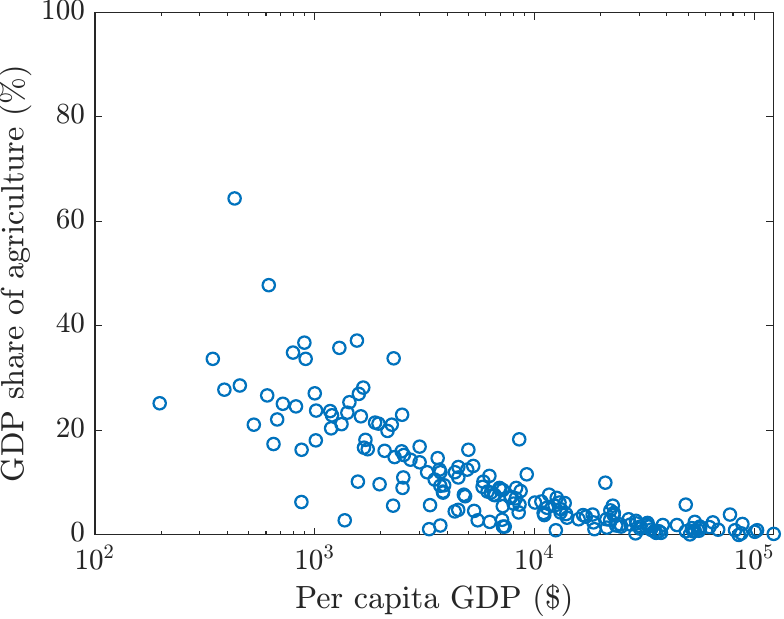}
    \includegraphics[width=0.48\linewidth]{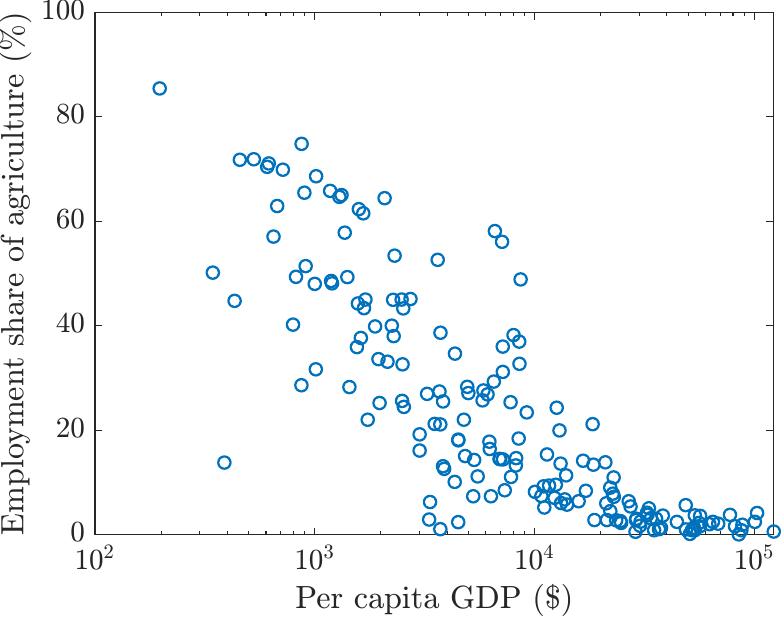}
    \caption{GDP and employment share of agriculture across countries in 2023.}
    \label{fig:country}
\end{figure}

The decline in the importance of land as a factor of production is partly due to biological constraints regarding the amount of food people can consume (where land produces agricultural products) or the amount of leisure time people can spend (where land produces amenities like tennis courts and national parks). Although people living in modern capitalistic societies have tremendously benefited from technological progress over the past decades such as the development of computers, Internet, smartphones, and electric vehicles, introspection suggests that our dining and outdoor experiences---the quality of ``land-intensive products''---have not changed much.

At the same time, land also plays a significant role as a scarce means of savings. Figure \ref{fig:OECD}, which reproduces \citet[Figure 2.1]{OECD2022},\footnote{Data are available at \url{https://stat.link/z6oj0i}.} shows that among 29 OECD countries, real estate (owner-occupied housing and secondary real estate) comprises more than 50\% of household wealth in 27 countries. Land possesses a few characteristics that make it suitable as a store of value. First, unlike cryptocurrency, land has an intrinsic value because it can be used as a factor of production in agriculture, construction, housing, and leisure. Second, unlike gold (which is chemically homogeneous), each land parcel is immobile and unique and hence property rights are well-defined, which makes it difficult to steal. Third, relative to durable goods such as vehicles, land is more durable as it cannot be destroyed absent natural disasters, sea level rise, and pollution.

\begin{figure}[!htb]
    \centering
    \includegraphics[width=0.7\linewidth]{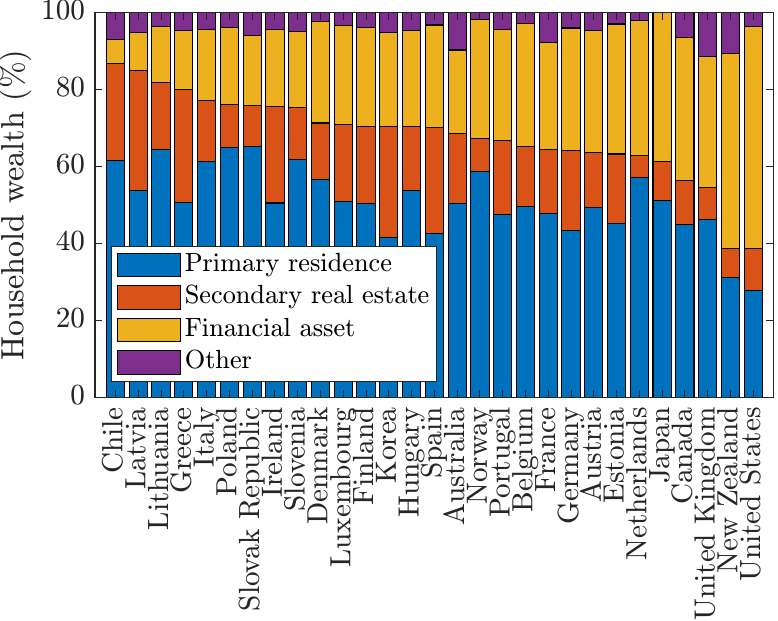}
    \caption{Composition of gross household wealth in OECD countries.}
    \label{fig:OECD}
\end{figure}

This paper theoretically studies the long-run behavior of land prices in modern economies where the importance of land as a factor of production diminishes, yet, land remains to play a significant role as a store of value. In a plausible economic model with land and aggregate risk, we establish a theorem showing the tight link between unbalanced (uneven) productivity growth, elasticity of substitution between production factors, and overvaluation of land, meaning that the equilibrium land price exceeds its fundamental value defined by the present value of land rents.

To derive these findings, we consider a two-period stochastic overlapping generations (OLG) model with land and establish the Land Overvaluation Theorem. We identify economic conditions under which land overvaluation will necessarily emerge in equilibrium. Let us denote the labor and land productivities at time $t$ by $A_{Ht}$ and $A_{Xt}$, respectively. Let us also denote by $\sigma$ (a lower bound of) the elasticity of substitution between land and labor at sufficiently high input levels and assume $\sigma>1$. The main result of this paper, Theorem \ref{thm:main}, shows that if
\begin{equation*}
    \E_0\sum_{t=1}^\infty (A_{Ht}/A_{Xt})^{1/\sigma-1}<\infty,
\end{equation*}
then land is overvalued in equilibrium. Noting that $\sigma>1$ and hence $1/\sigma-1<0$, this land overvaluation condition holds whenever labor productivity $A_{Ht}$ grows faster than land productivity $A_{Xt}$ in the long run, \ie, unbalanced growth occurs. The intuition is as follows. Because the young save a fraction of income by purchasing land, along the equilibrium path, the land price increases together with wages, whose growth rate will be the same as labor productivity growth. On the other hand, the growth rate of land rents will be suppressed if the elasticity of substitution between land and labor exceeds 1, in which case land prices, pulled up by growing incomes, grow faster than land rents and therefore land will be overvalued, with an upward trend in the price-rent ratio.\footnote{To be precise, the ratio increases over time depending on whether the elasticity of substitution exceeds 1 at high input levels, not necessarily globally. In \S\ref{subsec:general_ES}, we justify our assumption of $\sigma>1$ in several ways based on both empirical and theoretical grounds.} Moreover, since the equilibrium is unique, land overvaluation necessarily emerges.

There are three important implications to be drawn from our Land Overvaluation Theorem. First, our analysis illustrates the key mechanism of how land overvaluation emerges, where unbalanced growth and elasticity of substitution play a crucial role. Second, unlike the usual perspective on land overvaluation (sometimes called land bubbles) as short-run phenomena with boom-bust cycles, our analysis shows the emergence of land overvaluation on the long-run trend along economic development. In reality, as economies develop, structural transformation occurs from the land-intensive agricultural economy to the labor- or knowledge-intensive economy. During this transition, while the importance of land as a factor of production diminishes, as long as land remains important as a store of value, land necessarily becomes overvalued. To illustrate this point clearly, in \S\ref{subsec:two-sector}, we present a two-sector model as a special case and show that land overvaluation will occur as the unique equilibrium outcome in the major historical trend of the transition from a Malthusian economy to a modern economy. In addition to the shift from the agricultural economy to the labor- or knowledge-intensive economy, land use is also changing. In \S\ref{subsec:urban}, we present another special case in which land overvaluation emerges with urban formation with unbalanced growth. Third, our theorem in an economy with aggregate risk also provides a new insight on short-term fluctuations that deviate from the long-run trend. When productivities fluctuate, so do land prices. In standard asset pricing models, these valuations and fluctuations always reflect fundamentals. In contrast, our analysis shows that land is always overvalued in the long-run trend, associated with expansions and contractions in the size of overvaluation over short periods of time that may appear to be the emergence and collapse of large land bubbles. Our model provides a theoretical foundation for recurrent stochastic bubbles. We emphasize that this insight can only be obtained in an economy with aggregate risk. In a deterministic economy, if the land price-rent ratio converges, it implies that the land price reflects fundamentals. In contrast, in an economy with aggregate risk, even if the ratio appears to be stable for an extended period of time, it does not necessarily mean land prices reflect fundamentals. So long as the condition satisfying the Theorem is satisfied, the price of land is always overvalued, and the extent of the overvaluation is constantly changing. \S\ref{subsec:recurrent} illustrates these points by considering that productivity changes in labor and land stochastically evolve according to a Markov chain.

\section{Related literature}\label{subsec:literature}
We employ a standard two-period OLG model with land as in \citet{McCallum1987,Rhee1991,HansenPrescott2002,Mountford2004,Stiglitz2015}, where land plays the dual role of factor of production and store of value. Our paper is different because we focus on asset pricing, unbalanced growth, and land overvaluation. \citet{Rhee1991} shows the importance of elasticity of substitution exceeding one for the possibility of dynamic inefficiency.\footnote{Note that some of the analysis in \citet{Rhee1991} is loose, as his Assumption A involves endogenous objects and the proof is claimed to be equivalent to \citet{Tirole1985}, which also involves high-level assumptions as discussed in \citet[\S5.2]{HiranoToda2024JME}.} In contrast, we show the importance of this condition for the necessity of land overvaluation but dynamic inefficiency is irrelevant.\footnote{Following \citet{Diamond1965,Tirole1985}, it is often noted that dynamic inefficiency is important for the existence of asset bubbles. A typical approach considers a hypothetical economy in which there is no dividend-paying asset like land, and then shows that under some conditions, dynamic inefficiency can arise, in which case pure bubble assets can circulate. However, once we introduce land from the beginning, under some conditions land bubbles emerge as the unique equilibrium outcome, so it is meaningless to think about a bubbleless world.}

There is also a literature that emphasizes the role of land as collateral \citep{KiyotakiMoore1997,Krishnamurthy2003,Lorenzoni2008,Kocherlakota2013}. We abstract from this role to clarify the mechanism of how land overvaluation necessarily occurs in a bare-bones model where land plays the dual role of production factor and store of value. While land prices reflect fundamentals in \citet{KiyotakiMoore1997,Lorenzoni2008}, land is intrinsically useless like fiat money or cryptocurrencies and may have a positive value in \citet{Kocherlakota2013}, so land can be overvalued. In contrast, in our model, land is a productive asset used as an input for production, in which case land is shown to be inevitably overvalued. We identify the underlying economic conditions.

As in \citet{Lucas1978} and the large subsequent literature, we study asset pricing in an economy with aggregate uncertainty. In macro-finance, it is well known since \citet{SantosWoodford1997} that there is a fundamental difficulty in generating asset overvaluation (sometimes called asset bubbles) in dividend-paying assets.\footnote{\label{fn:difficulty}Bubbles cannot arise in rational equilibrium models if dividends comprise a non-negligible fraction of aggregate endowments. See \citet[Theorem 3.3]{SantosWoodford1997} and \citet[\S3.4]{HiranoToda2024JME} for details.} Perhaps due to this difficulty, progress in macro-finance models that describe realistic asset overvaluation in land, housing, and stocks has been slow. Indeed, the literature on rational bubbles has almost exclusively focused on pure bubbles like fiat money. However, applications to policy, quantitative, and empirical analyses using pure bubble models are severely limited due to the following reasons. First, in reality, it is hard to find pure bubble assets other than fiat money or cryptocurrencies. Second, there exist a continuum of bubble equilibria as well as fundamental equilibria, which makes model predictions non-robust.\footnote{This statement is often true but not always, as  \citet{Scheinkman1980,Santos1990} provide counterexamples of indeterminacy in endowment economies. \citet{HiranoToda2024EL} examine if their result also holds in production economies and prove that there exist a continuum of monetary equilibria.} Third, the price-dividend ratio, which is used as a bubble detector in the econometric literature \citep{PhillipsShiYu2015,PhillipsShi2018}, cannot even be defined. Hence, connecting pure bubble models to the bubble detection literature is impossible. \citet[\S7]{Wilson1981}, \citet[\S6.1.2]{LeVanPham2016}, and \citet[\S4.1]{BosiLeVanPham2022} provide examples of asset price bubbles when dividends shrink relative to the economy. \citet{Tirole1985,BosiHa-HuyLeVanPhamPham2018} study asset price bubbles in an OLG model with capital accumulation. In a recent paper, \citet{HiranoToda2025JPE} identify the economic conditions under which asset price bubbles attached to dividend-paying assets necessarily emerge and establish the concept of the necessity of bubbles in modern macro-finance models including Bewley-type infinite-horizon models.\footnote{For more discussion of these contributions, see \citet[\S5]{HiranoToda2024JME}.} Since land in our model is used as a factor of production yielding positive rents, land may be interpreted as a Lucas tree with endogenous dividends. We build on the insight of \citet{HiranoToda2025JPE} but the crucial differences are that we make the land rent \emph{endogenous} and allow for \emph{aggregate uncertainty}.

Concerning unbalanced growth, \citet{Baumol1967} points out the implications for economic development when different sectors have different productivity growth rates. \citet{HansenPrescott2002} consider a two-sector OLG model with uneven productivity growth rates across the capital-intensive (Solow) sector and the land-intensive (Malthus) sector and argue that land becomes unimportant as a factor of production as the economy develops. \citet{AcemogluGuerrieri2008} show in a two-sector model that differences in factor proportions across different sectors combined with capital deepening leads to unbalanced growth. The elasticity of substitution between the two sectors play a key role for growth dynamics. \citet{Matsuyama1992,BueraKaboski2012,Boppart2014,FujiwaraMatsuyama2024} use non-homothetic preferences to generate unbalanced growth. A crucial difference between our work and this literature is that we show the tight theoretical link between unbalanced growth, elasticity of substitution, and land overvaluation, while the literature abstracts from asset pricing. In contrast, we clarify the mechanism of how unbalanced growth dynamics plays a key role in generating asset overvaluation. Finally, our example in \S\ref{subsec:urban} is also related to the literature on urban formation \citep{Krugman1991, GlaeserHediScheinkmanShleifer1992}. Obviously, urban formation is characterized by unbalanced growth. While this literature abstracts from asset pricing, our Theorem can be applied to study urban land prices.

\section{Fundamental value and bubble}

Our paper studies land overvaluation. To this end, we need to define the fundamental value of assets and asset price bubbles. Although the definitions are standard, which we discussed elsewhere \citet{HiranoToda2024JME,HiranoToda2025JPE}, to make the paper self-contained, we present the definitions in a general setting following \citet{HiranoTodaClarification}.

Consider an infinite-horizon economy with a homogeneous good and time indexed by $t=0,1,\dotsc$. Consider an asset with infinite maturity that pays dividend $D_t\ge 0$ and trades at ex-dividend price $P_t$, both in units of the time-$t$ good. Letting $m_{t\to t+1}$ be the stochastic discount factor (SDF) of an agent holding the asset at time $t$, the no-arbitrage asset pricing equation is given by
\begin{equation}
    P_t = \E_t[m_{t\to t+1}(P_{t+1}+D_{t+1})], \label{eq:noarbitrage}
\end{equation}
where $\E_t[\cdot]$ denotes the expectation conditional on time $t$ information. Define the state price deflator by $\pi_0=1$ and $\pi_t=\prod_{s=0}^{t-1}m_{s\to s+1}$ for $t\ge 1$. Multiplying $\pi_t$ to both sides of \eqref{eq:noarbitrage} and using the definition of $\pi_t$, we obtain
\begin{equation*}
    \pi_t P_t = \E_t[\pi_{t+1}(P_{t+1}+D_{t+1})].
\end{equation*}
Iterating forward yields
\begin{equation}
    \pi_tP_t=\E_t\sum_{s=t+1}^T \pi_sD_s+\E_t[\pi_TP_T]. \label{eq:P_iter}
\end{equation}
Because all terms are nonnegative, the sum in \eqref{eq:P_iter} from $s=t+1$ to $s=T$ is increasing in $T$ and bounded above by $\pi_tP_t$, so it converges almost surely as $T\to\infty$. Therefore the \emph{fundamental value} of the asset (the present value of dividends)
\begin{equation}
    V_t\coloneqq \frac{1}{\pi_t}\E_t\sum_{s=t+1}^\infty \pi_sD_s \label{eq:Vt}
\end{equation}
is well-defined. Letting $T\to\infty$ in \eqref{eq:P_iter}, we obtain $P_t=V_t+B_t$, where we define the \emph{asset price bubble} as
\begin{equation}
    B_t\coloneqq \lim_{T\to\infty} \frac{1}{\pi_t}\E_t[\pi_TP_T]\ge 0.\label{eq:Bt}
\end{equation}
Two remarks are in order. First, the economic meaning of the bubble component $B_t$ in \eqref{eq:Bt} is that it captures a speculative aspect, that is, agents buy the asset now for the purpose of resale in the future, rather than for the purpose of receiving dividends. When $\lim_{T\to\infty}\E_t[\pi_TP_T]=0,$ the aspect of speculation becomes negligible and asset prices are determined only by factors that are backed in equilibrium, namely future dividends. On the other hand, if $\lim_{T\to\infty}\E_t[\pi_TP_T]>0$, equilibrium asset prices contain a speculative aspect. Second, although the definition is the same, there is a discontinuity in proving the existence of a bubble between the cases with $D_t=0,$ \ie, fiat money or cryptocurrencies, and $D_t>0$ such as land, housing or stocks. In other words, as noted in Footnote \ref{fn:difficulty}, there is a fundamental difficulty in generating a bubble attached to an asset with $D_t>0$.

\section{Land Overvaluation Theorem}\label{sec:general}

In what follows, following \citet{Baumol1967}, we refer to a situation with uneven productivity growth between different production factors or different sectors as ``unbalanced growth''. In this section, we uncover the mechanism of how land overvaluation necessarily emerges with unbalanced growth dynamics. We also highlight the role of the elasticity of substitution between land and other production factors.

\subsection{Model}

We consider a stochastic two-period OLG model. Uncertainty is resolved according to a filtration $\set{\mathcal{F}_t}_{t=0}^\infty$ over a probability space $(\Omega,\mathcal{F},P)$. We denote conditional expectations by $\E_t[\cdot]=\E[\cdot \mid \mathcal{F}_t]$.

\paragraph{Preferences}

At each time $t$, a unit mass of agents are born, who live for two periods and derive utility
\begin{equation}
    (1-\beta)\log c_t^y+\beta \E_t[\log c_{t+1}^o]\label{eq:CD}
\end{equation}
from consumption $(c_t^y,c_{t+1}^o)$ when young and old, where $\beta\in (0,1)$ dictates time preference.\footnote{The logarithmic (Cobb-Douglas) utility \eqref{eq:CD} is convenient for obtaining closed-form solutions. This assumption is not essential, as shown by \citet[Theorem 2]{HiranoToda2025JPE}.} Each period, the young are endowed with one unit of labor, while the old have none. At $t=0$, there is a unit mass of initial old agents who only care about their consumption $c_0^o$. The initial old is endowed with a unit supply of land, which is durable and non-reproducible.

\paragraph{Technologies}

There are two factors of production, labor (denoted by $H$) and land (denoted by $X$). Although we call $H$ ``labor'', it may be interpreted as human or intangible capital in a modern economy. To illustrate the point of how unbalanced growth necessarily generates land overvaluation, we focus on human capital and abstract from physical capital such as buildings and machinery.\footnote{For a model in which both physical capital and land serve as a store of value, see \citet[\S6]{HiranoTodaReal}.} As we will show below, the point that different factors for production have different productivity growth rates plays a crucial role, as well as the elasticity of substitution between them.

Without loss of generality, we only specify the aggregate production function, as it is well known that if each sector or firm is competitive and markets are frictionless, profit maximization at the individual and aggregate levels are equivalent. (See the example in \S\ref{subsec:two-sector} below.) Below, we say that a production function $F(H,X)$ is \emph{neoclassical} if $F:\R_{++}^2\to \R_{++}$ is homogeneous of degree 1, concave, continuously differentiable, and satisfies $F_H,F_X>0$.

\begin{asmp}\label{asmp:F}
The time $t$ aggregate production function takes the form
\begin{equation*}
	F_t(H,X)=F(A_{Ht}H,A_{Xt}X),
\end{equation*}
where $F$ is a neoclassical production function and $A_{Ht},A_{Xt}>0$ are $\mathcal{F}_t$-measurable factor-augmenting productivities.
\end{asmp}

\paragraph{Equilibrium}

Let $w_t>0$ be the wage, $r_t>0$ the land rent, and $P_t>0$ the land price (excluding rent), all in units of the current consumption good. Generation $t$ seeks to maximize the utility \eqref{eq:CD} subject to the budget constraints
\begin{subequations}\label{eq:budget_OLG}
\begin{align}
    &\text{Young:} & c_t^y+P_tx_t&=w_t, \label{eq:budget_young}\\
    &\text{Old:} & c_{t+1}^o&=(P_{t+1}+r_{t+1})x_t, \label{eq:budget_old}
\end{align}
\end{subequations}
where $x_t$ is the demand for land. The definition of equilibrium is standard.

\begin{defn}\label{defn:eq}
A \emph{rational expectations equilibrium} consists of adapted processes of prices $\set{(w_t,r_t,P_t)}_{t=0}^\infty$ and allocations $\set{(c_t^y,c_t^o,x_t,H_t,X_t)}_{t=0}^\infty$, such that, for each $t$,
\begin{enumerate}
    \item (Utility maximization) $(c_t^y,c_{t+1}^o,x_t)$ maximizes the utility \eqref{eq:CD} subject to the budget constraints \eqref{eq:budget_OLG},
    \item (Profit maximization) $(H_t,X_t)$ maximizes the profit $F_t(H_t,X_t)-w_tH_t-r_tX_t$,
    \item (Market clearing) $H_t=1$, $X_t=1=x_t$, and $c_t^y+c_t^o=F_t(H_t,X_t)$.
\end{enumerate}
\end{defn}

Due to log utility, the existence and uniqueness of equilibrium are immediate. Here and elsewhere, we denote partial derivatives using subscripts, \eg, $F_H\coloneqq \partial F/\partial H$. Below, all proofs are deferred to Appendix \ref{sec:proof}.

\begin{prop}\label{prop:eq}
If Assumption \ref{asmp:F} holds, then the economy has a unique equilibrium, which is characterized by the following equations.
\begin{subequations}\label{eq:eq_obj}
\begin{align}
    &\text{Wage:} & w_t&=F_H(A_{Ht},A_{Xt})A_{Ht}, \label{eq:w}\\
    &\text{Rent:} & r_t&=F_X(A_{Ht},A_{Xt})A_{Xt}, \label{eq:r}\\
    &\text{Land price:} & P_t&=\beta w_t, \label{eq:P}\\
    &\text{Young consumption:} & c_t^y&=(1-\beta)w_t, \label{eq:y}\\
    &\text{Old consumption:} & c_t^o&=\beta w_t+r_t. \label{eq:z}
\end{align}
\end{subequations}
\end{prop}

\subsection{Elasticity of substitution}\label{subsec:general_ES}

The elasticity of substitution turns out to play a crucial role in generating land overvaluation. Recall that the elasticity of substitution $\sigma$ between production factors is defined by the percentage change in relative factor inputs with respect to the percentage change in relative factor prices
\begin{equation}
    \sigma=-\frac{\partial \log (H/X)}{\partial \log (w/r)}, \label{eq:ES}
\end{equation}
where the derivative is taken along the production possibility frontier $F(H,X)=\text{constant}$. (See Appendix \ref{subsec:proof_ES} for details.) For neoclassical production functions, the elasticity of substitution is given by $\sigma_F=(F_HF_X)/(FF_{HX})$ (Lemma \ref{lem:sigma}).

To derive asset pricing implications, we restrict $\sigma_F$ as follows.

\begin{asmp}\label{asmp:ES}
The elasticity of substitution of the neoclassical production function $F$ exceeds 1 at high input levels:
\begin{equation*}
    \liminf_{H\to\infty}\sigma_F(H,1)>\sigma>1.
\end{equation*}
\end{asmp}

We justify Assumption \ref{asmp:ES} in several ways.

The first justification is empirical. \citet*{EppleGordonSieg2010} find that the elasticity of substitution between land and non-land factors for producing housing service is 1.16 for residential properties and 1.39 for commercial properties in Allegheny County, Pennsylvania. \citet{AhlfeldtMcMillen2014} argue that the estimation approach of \citet*{EppleGordonSieg2010} is less susceptible to measurement error than older estimates, which are likely biased downwards. They find that the elasticity of substitution is around 1.25 for Chicago and Berlin. 

The second justification is the pathological behavior of interest rates with $\sigma<1$. To see this, suppose for simplicity that the production function exhibits constant elasticity of substitution (CES)
\begin{equation}
    F(H,X)=\left(\alpha H^{1-\rho}+(1-\alpha)X^{1-\rho}\right)^\frac{1}{1-\rho}, \label{eq:CES}
\end{equation}
where $\sigma=1/\rho$ is the elasticity of substitution and $\alpha\in (0,1)$. Suppose productivities are given by
\begin{equation}
    (A_{Ht},A_{Xt})=(G_H^t,G_X^t) \label{eq:AHX}
\end{equation}
with $G_H>G_X$, and assume $\sigma<1$. $G_H>G_X$ means that the productivity growth rate of human capital is greater than that of land. Using \eqref{eq:eq_obj}, we can bound the gross risk-free rate from below as
\begin{equation*}
    R_t=\frac{\beta w_t+r_t}{\beta w_{t-1}}\ge \frac{r_t}{\beta w_{t-1}}.
\end{equation*}
Using \eqref{eq:w} \eqref{eq:r}, \eqref{eq:CES}, \eqref{eq:AHX}, and $\rho>1$, we obtain
\begin{align*}
    w_t&=\alpha\left(\alpha  G_H^{(1-\rho)t}+(1-\alpha)G_X^{(1-\rho)t}\right)^\frac{\rho}{1-\rho}G_H^{(1-\rho)t}\\
    &\sim \alpha(1-\alpha)^\frac{\rho}{1-\rho}G_H^{(1-\rho)t}G_X^{\rho t},\\
    r_t&=(1-\alpha)\left(\alpha  G_H^{(1-\rho)t}+(1-\alpha)G_X^{(1-\rho)t}\right)^\frac{\rho}{1-\rho}G_X^{(1-\rho)t}\\
    &\sim (1-\alpha)^\frac{1}{1-\rho}G_X^t.
\end{align*}
Therefore
\begin{equation*}
    R_t\ge \frac{r_t}{\beta w_{t-1}}\sim \frac{1-\alpha}{\alpha\beta}G_X(G_H/G_X)^{(\rho-1)(t-1)}\to \infty
\end{equation*} 
because $G_H>G_X$ and $\rho>1$. An interest rate diverging to infinity is counterfactual and pathological.

\subsection{Unbalanced growth and land overvaluation}

We now establish Land Overvaluation Theorem as the main result of this paper.

\begin{thm}[Land Overvaluation]\label{thm:main}
Suppose Assumptions \ref{asmp:F}, \ref{asmp:ES} hold and
\begin{equation}
    \E_0\sum_{t=1}^\infty (A_{Ht}/A_{Xt})^{1/\sigma-1}<\infty\label{eq:cond_A}
\end{equation}
almost surely. Then land is overvalued in equilibrium.
\end{thm}

The condition \eqref{eq:cond_A} can be understood as follows. Suppose for simplicity that $(A_{Ht},A_{Xt})=(G_H^t,G_X^t)$, so productivity growth is exponential. Then the $t$-th term in the sum \eqref{eq:cond_A} is $(G_H/G_X)^{(1/\sigma-1)t}$, which is summable if $\sigma>1$ and $G_H>G_X$. Thus condition \eqref{eq:cond_A} roughly says that labor productivity growth is higher than land productivity growth in the long run. The intuition for Theorem \ref{thm:main} is similar to that noted in the introduction, so we do not repeat it. It is important to note that since the equilibrium is unique by Proposition \ref{prop:eq}, under the conditions in Theorem \ref{thm:main}, there are no equilibria in which the land price equals its fundamental value. In other words, in this economy with aggregate uncertainty, the only possible equilibrium is one that features land overvaluation.\footnote{By a symmetric argument, if the elasticity of substitution is bounded \emph{above} by $\sigma<1$ and $G_H<G_X$, a land bubble also inevitably occurs. If $\sigma=1$ or $G_H=G_X$, then bubbles do not occur, but these cases are obviously knife-edge.}

In our model, land is the primary store of value. A natural question is what happens if there are multiple assets that serve as a store of value such as stocks, gold and cryptocurrency. In this case, the extent of overvaluation in individual assets could be indeterminate. However, the aggregate amount of overvaluation and the equilibrium outcome are determinate and identical, as in the present model. Hence, from a macroeconomic perspective, this equilibrium indeterminacy in individual assets does not matter. This is basically the same as the ``bubble substitution'' argument in \citet[\S5]{Tirole1985}; see also \citet[\S6]{HiranoTodaReal} for an example of overvaluation in the aggregate value of stocks and land. Nonetheless, we emphasize the role of land as a means of savings because as noted in the introduction, land is a focal point due to its characteristics.

Theorem \ref{thm:main} has three important implications. First, it clarifies the role of unbalanced growth and elasticity of substitution for generating land overvaluation. Second, we can derive a new insight on the long-run behavior of land prices in a modern economy. The conventional view is that on the long-run trend, the land price should reflect its fundamental value, even if it may deviate from the fundamental value temporarily. In sharp contrast with this widely-held view, Theorem \ref{thm:main} implies that during the process of economic development characterized by unbalanced productivity growth, land overvaluation will necessarily arise. We discuss the third implication of Theorem \ref{thm:main} by specializing it in the next section.

\section{Examples}\label{sec:example}

In this section, we illustrate Theorem \ref{thm:main} by discussing three special cases.

\subsection{Land price from Malthusian to modern economy}\label{subsec:two-sector}

We present a simple example to illustrate the inevitability of land overvaluation in the context of a major historical trend involving structural changes. To demonstrate this, we employ the two-sector model of \citet{HansenPrescott2002}, simplified without capital. This model can be interpreted as capturing how land prices are determined during the transition from a land-intensive Malthusian economy to a knowledge-intensive modern economy. It is also similar to the example in \citet[\S3.2]{HiranoToda2025JPE} but extended to include aggregate uncertainty and an endogenous structural change.

\paragraph{Model}

Instead of an aggregate production function as in \S\ref{sec:general}, we suppose there are two production sectors denoted by $j=1,2$. Sector 1 is a traditional land-intensive industry where both labor and land are used as inputs for production, such as agriculture and construction. Sector 2 is a modern industry where labor (human capital) is the primary input for production, such as technology, finance, and information and communication. The time $t$ production function of sector $j$ is $F_{jt}(H,X)$. For simplicity, we suppose that technologies in sectors 1 and 2 are Cobb-Douglas and linear, respectively:
\begin{subequations}\label{eq:F_OLG}
\begin{align}
    F_{1t}(H,X)&=A_{1t}H^\alpha X^{1-\alpha},\\
    F_{2t}(H,X)&=A_{2t}H,
\end{align}
\end{subequations}
where $A_{jt}>0$ denotes the productivity in sector $j$ at time $t$ and $\alpha\in (0,1)$.

\paragraph{Equilibrium}

As usual, an equilibrium is defined by utility maximization, profit maximization, and market clearing. We omit the formal definition as it is similar to Definition \ref{defn:eq}. Since in equilibrium, we have $X_{1t}=1$ and $X_{2t}=0$, using the functional form of the production functions \eqref{eq:F_OLG}, profit maximization implies the first-order conditions
\begin{subequations}\label{eq:wr}
    \begin{align}
        w_t&=\alpha A_{1t}H_{1t}^{\alpha-1}, \label{eq:w1}\\
        r_t&=(1-\alpha)A_{1t}H_{1t}^\alpha, \label{eq:rent}\\
        w_t&\ge A_{2t},~\text{with equality if $H_{2t}>0$}. \label{eq:w2}
    \end{align}
\end{subequations}

There are two cases to consider. If $H_{2t}>0$, then \eqref{eq:w2} implies $w_t=A_{2t}$, and \eqref{eq:w1} then implies $H_{1t}=(\alpha A_{1t}/A_{2t})^\frac{1}{1-\alpha}$. For $H_{2t}>0$, we need $H_{1t}<1$, or equivalently $A_{2t}>\alpha A_{1t}$. Otherwise, we have $(H_{1t},H_{2t})=(1,0)$ and $w_t=\alpha A_{1t}$ by \eqref{eq:w1}. Combining the two cases and using \eqref{eq:wr}, we always have
\begin{subequations}
\begin{align}
    w_t&=\max\set{\alpha A_{1t},A_{2t}}, \label{eq:wt}\\
    H_{1t}&=\min\set{\alpha A_{1t}/A_{2t},1}^\frac{1}{1-\alpha}, \label{eq:H1t}\\
    r_t&=(1-\alpha)A_{1t}\min\set{\alpha A_{1t}/A_{2t},1}^\frac{\alpha}{1-\alpha}. \label{eq:rt}
\end{align}
\end{subequations}

\paragraph{Unbalanced growth and land overvaluation}

Although this two-sector model appears rather different from the model in \S\ref{sec:general}, the former is actually a special case of the latter.

As is well known, profit maximization at the individual sector or firm level is equivalent to that at the aggregate level. Consider the aggregation of the two production functions in \eqref{eq:F_OLG}. Suppressing the $t$ subscript and setting $(X_1,X_2)=(X,0)$, it suffices to solve
\begin{equation*}
    F(H,X)\coloneqq
    \max\set{\sum_{j=1}^2 F_j(H_j,X_j):\sum_{j=1}^2 H_j=H,\sum_{j=1}^2 X_j=X}.
\end{equation*}
Applying the Karush-Kuhn-Tucker theorem, it is straightforward to show that the aggregate production function is\begin{equation*}
    F_t(H,X)=\begin{cases*}
        A_{1t}H^\alpha X^{1-\alpha} & if $\alpha A_{1t}/A_{2t}\ge (H/X)^{1-\alpha}$,\\
        A_{2t}H+(1-\alpha)\alpha^\frac{\alpha}{1-\alpha}(A_{1t}/A_{2t}^\alpha)^\frac{1}{1-\alpha}X & if $\alpha A_{1t}/A_{2t}<(H/X)^{1-\alpha}$.
    \end{cases*}
\end{equation*}
Therefore if we define
\begin{equation*}
    F(H,X)=\min\set{H^\alpha X^{1-\alpha},H+(1-\alpha)\alpha^\frac{\alpha}{1-\alpha}X}
\end{equation*}
and $(A_{Ht},A_{Xt})=(A_{2t},(A_{1t}/A_{2t}^\alpha)^\frac{1}{1-\alpha})$, we obtain $F_t(H,X)=F(A_{Ht}H,A_{Xt}X)$ and Assumption \ref{asmp:F} holds. Furthermore, because
\begin{equation*}
    F(H,X)=H+(1-\alpha)\alpha^\frac{\alpha}{1-\alpha}X
\end{equation*}
is linear for $H\ge \alpha^\frac{1}{1-\alpha}X$, in which case the elasticity of substitution is $\sigma_F=\infty$, Assumption \ref{asmp:ES} also holds. The relative productivity is $A_{Ht}/A_{Xt}=(A_{2t}/A_{1t})^\frac{1}{1-\alpha}$, so condition \eqref{eq:cond_A} (with $\sigma=\infty$) reduces to
\begin{equation}
    \E_0\sum_{t=1}^\infty (A_{1t}/A_{2t})^\frac{1}{1-\alpha}<\infty. \label{eq:bubble_cond}
\end{equation}

This result implies that a land price bubble emerges if and only if the productivity growth in sector 2 is sufficiently faster than that in sector 1 so that $(A_{1t}/A_{2t})^\frac{1}{1-\alpha}$ is summable. Thus, a land price bubble emerges with economic development. For instance, suppose productivities grow exponentially, so $(A_{1t},A_{2t})=(A_1G_1^t,A_2G_2^t)$. Then the land overvaluation condition \eqref{eq:bubble_cond} is equivalent to $G_2>G_1$, \ie, \emph{unbalanced growth}. When the economy features multiple sectors as in reality, there is no reason to expect equal growth rates across sectors. The slightest introspection suggests that it would be a miracle if the rate of technological progress were the same in 19th century trains and (horse-drawn) carriages, 20th century computers and calculators, or early 21st century electric vehicle batteries and internal combustion engines. Unbalanced growth is a natural and general feature in the process of economic development.

When the productivity in sector 2 is so low that $A_{2t}\le \alpha A_{1t}$, or $t\le t^*\coloneqq \log (\alpha A_1/A_2)/\log (G_2/G_1)$ in the exponential example, by \eqref{eq:H1t} we have $H_{1t}=1$, so all labor is hired in the land-intensive sector. If this situation is expected to continue, by \eqref{eq:rt} and $P_t=\beta w_t$, both land rents and prices are expected to grow at the same rate of $G_1$. In the condition \eqref{eq:bubble_cond}, since $A_{1t}>A_{2t},$  $\E_0\sum_{t=1}^\infty (A_{1t}/A_{2t})^\frac{1}{1-\alpha}=\infty$. Therefore, the land price reflects fundamentals. We may interpret this situation as a Malthusian economy dominated by land-intensive sectors. From this situation, as the productivity of sector 2 improves relative to sector 1 and $A_{2t}>\alpha A_{1t}$, or after the critical time $t^*$, sector 2 becomes active, \ie, an endogenous structural change occurs. We may interpret this situation as the Industrial Revolution or the rise of modern economy where human capital plays an important role in driving economic growth. If this relative productivity growth is expected to continue, the sum in \eqref{eq:bubble_cond} becomes finite, and a land bubble inevitably emerges. 

Moreover, when sector 2 is active together with the structural change, as the relative productivity $A_{2t}/A_{1t}$ increases, by \eqref{eq:H1t} the employment share of sector 1, $H_{1t}=(\alpha A_{1t}/A_{2t})^\frac{1}{1-\alpha}$, declines. Thus, our model is consistent with stylize facts documented in Figures \ref{fig:US}--\ref{fig:OECD} that the employment share of land-intensive sectors has been declining but land remains to be an important store of value.

From this analysis, we can draw an important insight. That is, in a stationary world in which only the land-intensive sector persists or the productivity growth rates in the two sectors are the same, land prices and rents grow at the same rate and therefore land price bubbles can never emerge. Once new sectors with higher productivity growth emerge and the economy shifts to a nonstationary world in which those new sectors drive economic growth, land price bubbles become inevitable.

\subsection{Recurrent stochastic bubbles}\label{subsec:recurrent}

In \S\ref{subsec:two-sector}, we have derived the insight that land overvaluation will occur in the major historical trend of the transition to the modern economy where human capital drives long-run economic growth. In this section, we connect this insight with short-run fluctuations.

The production function takes the CES form \eqref{eq:CES}. Let $A_t\coloneqq A_{Ht}/A_{Xt}$ be the relative productivity of labor. The state of the economy at time $t$ is denoted by $n_t$, which evolves over time according to a Markov chain with transition probability matrix $\Pi=(\pi_{nn'})$, where $\pi_{nn'}=\Pr(n_t=n' \mid n_{t-1}=n)$. The relative productivity $A_t$ evolves over time as a Markov multiplicative process
\begin{equation}
    A_t=G_tA_{t-1}, \label{eq:MMP}
\end{equation}
where $G_t$ conditional on $(n_{t-1},n_t)=(n,n')$ is an \iid copy of some random variable $G_{nn'}>0$.\footnote{See \citet[\S 2]{BeareToda2022ECMA} for more details on such Markov multiplicative processes.} Let $S_n(A)$ be the value of \eqref{eq:cond_A} when $(A_0,n_0)=(A,n)$. Due to the multiplicative nature of shocks and homogeneity, we may write $S_n(A)=s_nA^{\rho-1}$ for some constant $s_n>0$, where $\rho=1/\sigma$. A dynamic programming argument shows
\begin{equation}
    s_n=1+\sum_{n'=1}^N \pi_{nn'}\E[G_{nn'}^{\rho-1}]s_{n'}. \label{eq:sn}
\end{equation}
Defining the $N\times 1$ vector $s=(s_1,\dots,s_N)'$, the vector of ones $1=(1,\dots,1)'$, and the $N\times N$ nonnegative matrix $K=(\pi_{nn'}\E[G_{nn'}^{\rho-1}])$, we may rewrite \eqref{eq:sn} as
\begin{equation}
    s=1+Ks\iff s=(I-K)^{-1}1. \label{eq:s}
\end{equation}
A positive and finite solution to \eqref{eq:s} exists if and only if the spectral radius of $K$ (the maximum modulus of all eigenvalues) is less than 1.\footnote{This argument is similar to \citet{Toda2019JME,BorovickaStachurski2020}.} Therefore we obtain the following proposition.

\begin{prop}\label{prop:Markov}
Suppose the production function is CES with elasticity of substitution $\sigma>1$ and the relative labor productivity $A_t\coloneqq A_{Ht}/A_{Xt}$ follows the Markov multiplicative process \eqref{eq:MMP}. Let $K=(\pi_{nn'}\E[G_{nn'}^{1/\sigma-1}])$. Then land is overvalued if the spectral radius of $K$ is less than 1.
\end{prop}

As a numerical example, we set $\beta=0.5$, $\alpha=0.8$, $\sigma=1.25$, $N=2$, $\pi_{nn'}=1/3$ if $n\neq n'$, and $(G_{1n'},G_{2n'})=(1.1,0.95)$ for all $n'$, which implies that the spectral radius of $K$ is less than 1 and land is overvalued. Figure \ref{fig:sim} shows one simulation for 200 periods. The land price exhibits boom-bust cycles. The price-rent ratio steadily increases, consistent with Theorem \ref{thm:main}.

\begin{figure}[htb!]
\centering
\includegraphics[width=0.7\linewidth]{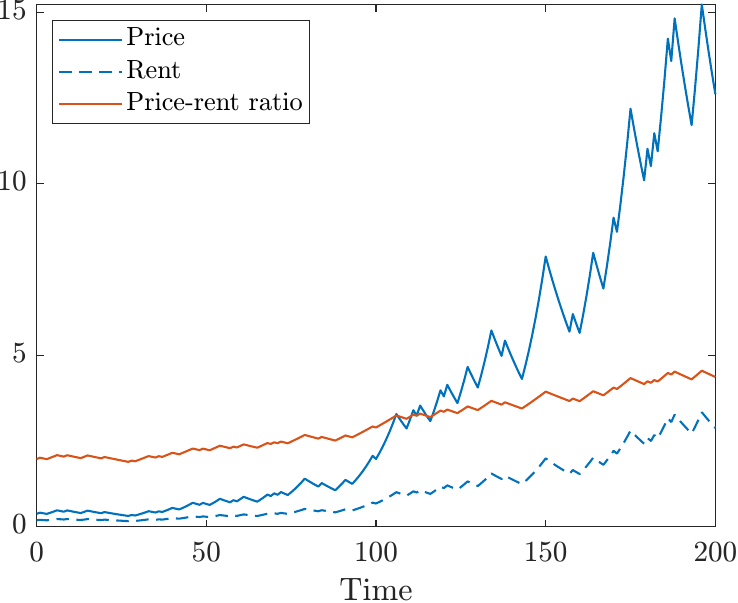}
\caption{Simulation of the numerical example of Proposition \ref{prop:Markov}.}\label{fig:sim}
\end{figure}

Proposition \ref{prop:Markov} and this numerical example provide the third implication of Theorem \ref{thm:main}. When productivities increase and remain to be high, land prices will continue to rise relative to the trend, which may look like an emergence of a large land price bubble. Conversely, if productivities decrease and remain to be so for an extended period of time, land prices will fall, which may appear to be a bursting of a land bubble. Thus, land prices exhibit recurrent booms and busts driven by fluctuations in productivities. Nonetheless, as long as the relative productivity growth of land is low, land will always be overvalued in the long-term trend, with the size of land overvaluation fluctuating over short periods of time and a steady upward trend in the price-rent ratio. Our model provides a theoretical foundation for recurrent stochastic bubbles. As noted in the introduction, this insight can only be obtained in an economy with aggregate risk.

\subsection{Land price with urban formation}\label{subsec:urban}

The example in \S\ref{subsec:two-sector} illustrates how land overvaluation emerges along the transition from a Malthusian (agricultural) economy to a modern economy. However, in modern economies, it seems that urban land, not agricultural land, has high value. In this section, we present a model with urban land overvaluation.

As before, we denote land by $X$ and non-land production factor by $H$. A fraction of non-land factor (\eg, construction workers) and land produce real estate $E=E(A_1\theta H,A_2X)$, where $E$ is a neoclassical production function, $A_1,A_2$ are factor-augmenting productivities, and $\theta\in (0,1)$ is the fraction of $H$ used in the real estate sector. The remaining $H$ (\eg, office workers) and real estate produce the consumption good using the Cobb-Douglas technology $Y=A_3((1-\theta)H)^\alpha E^{1-\alpha}$, where $A_3$ is total factor productivity and $\alpha\in (0,1)$. The final output given $(H,X)$ is then
\begin{equation}
    Y=A_3(1-\theta)^\alpha H^\alpha E(A_1\theta H,A_2X)^{1-\alpha}. \label{eq:Y}
\end{equation}

To reduce this model to a special case of the setting in \S\ref{sec:general}, we define the neoclassical production function
\begin{equation}
    F(H,X)\coloneqq H^\alpha E(H,X)^{1-\alpha}. \label{eq:FE}
\end{equation}
Then for arbitrary $A_H,A_X,\lambda>0$, we have
\begin{align}
    F(A_HH,A_XX)&=A_H^\alpha H^\alpha E(A_HH,A_XX)^{1-\alpha}\notag \\
    &=\frac{A_H^\alpha}{\lambda^{1-\alpha}}H^\alpha E(\lambda A_HH,\lambda A_XX)^{1-\alpha}. \label{eq:Flambda}
\end{align}
Comparing \eqref{eq:Y} and \eqref{eq:Flambda}, Assumption \ref{asmp:F} holds if we set
\begin{equation*}
    (A_H^\alpha \lambda^{\alpha-1},\lambda A_H,\lambda A_X)=(A_3(1-\theta)^\alpha,A_1\theta,A_2),
\end{equation*}
or equivalently $\lambda=(A_1\theta/(A_3(1-\theta)))^\alpha$ and $(A_H,A_X)=(A_1\theta/\lambda,A_2/\lambda)$. Note that the relative productivity
\begin{equation*}
    A_H/A_X=\theta A_1/A_2
\end{equation*}
depends only on $A_1/A_2$ and not on $A_3$. Therefore to apply Theorem \ref{thm:main}, it remains to characterize the elasticity of substitution of $F$.

\begin{prop}\label{prop:ES}
Let $E,F$ be the neoclassical production functions above with elasticity of substitution $\sigma_E,\sigma_F$. Then
\begin{equation}
    \sigma_F-1=\frac{1}{1+\alpha\frac{XE_X}{HE_H}\sigma_E}(\sigma_E-1). \label{eq:sigmaF}
\end{equation}
In particular, $\sigma_E>1$ implies $\sigma_F>1$.
\end{prop}

Empirical evidence suggests that the elasticity of substitution between land and non-land factors in producing real estate is above 1 \citep{EppleGordonSieg2010,AhlfeldtMcMillen2014}. Thus, if we assume $E$ is CES with elasticity of substitution $\sigma_E>1$, a straightforward calculation shows $(XE_X)/(HE_H)\propto (X/H)^{1-1/\sigma_E}\to 0$ as $H\to\infty$. Then \eqref{eq:sigmaF} implies
\begin{equation*}
    \lim_{H\to\infty}(\sigma_F-1)=\sigma_E-1>0,
\end{equation*}
so Assumption \ref{asmp:ES} holds and the conclusion of Theorem \ref{thm:main} remains valid if we replace the factor-augmenting productivities $(A_H,A_X)$ in the aggregate production function by the factor-augmenting productivities $(A_1,A_2)$ in the real estate sector. Thus, land overvaluation emerges in an urban economy if there is technological progress in the construction of real estate.

\appendix

\section{Proofs}\label{sec:proof}

\subsection{Proof of Proposition \ref{prop:eq}}
The first-order condition for profit maximization implies \eqref{eq:w} and \eqref{eq:r}. Define the return on land by
\begin{equation*}
    R_{t+1}=\frac{P_{t+1}+r_{t+1}}{P_t}.
\end{equation*}
Then the budget constraints \eqref{eq:budget_OLG} can be combined into one as
\begin{equation*}
    c_{t+1}^o=R_{t+1}(w_t-c_t^y).
\end{equation*}
Suppressing the time subscripts and substituting into the objective function, the young seek to maximize
\begin{equation*}
    (1-\beta)\log c^y+\E[\log c^o]=(1-\beta)\log c^y+\beta \log (w-c^y)+\beta\E[\log R].
\end{equation*}
Clearly this function is strictly concave in $c^y$ and achieves a unique maximum characterized by the first-order condition
\begin{equation*}
    \frac{1-\beta}{c^y}-\frac{\beta}{w-c^y}=0\iff c^y=(1-\beta)w,
\end{equation*}
which is \eqref{eq:y}. Since in equilibrium we have $x_t=1$, the land price satisfies $P_t=P_tx_t=w_t-y_t=\beta w_t$, which is \eqref{eq:P}. \hfill \qedsymbol

\subsection{Elasticity of substitution}\label{subsec:proof_ES}

A mathematically more convenient way to define the elasticity of substitution than \eqref{eq:ES} is the following. Let $h=\log(H/X)$ be the log relative inputs. Then noting that $w=F_H$ and $r=F_X$, \eqref{eq:ES} can be rewritten as
\begin{equation}
    \rho(H,X)\coloneqq \frac{1}{\sigma(H,X)}=-\frac{\partial \log (F_H/F_X)}{\partial h}, \label{eq:sigma}
\end{equation}
where we set $(H,X)=(X\e^h,X)$ to compute the derivative and substitute $h=\log(H/X)$.

\begin{lem}\label{lem:sigma}
Let $F$ be a neoclassical production function. Then its elasticity of substitution $\sigma_F(H,X)$ satisfies
\begin{equation}
	\sigma_F=\frac{F_HF_X}{FF_{HX}}. \label{eq:sigma_homog}
\end{equation}
\end{lem}

\begin{proof}
Since $F$ is homogeneous of degree 1, $F_H$ is homogeneous of degree 0. Therefore differentiating both sides of
\begin{align*}
	F(\lambda H,\lambda X)&=\lambda F(H,X),\\
	F_H(\lambda H,\lambda X)&=F_H(H,X)
\end{align*}
with respect to $\lambda$ and setting $\lambda=1$, we obtain
\begin{subequations}\label{eq:euler}
\begin{align}
	HF_H+XF_X&=F, \label{eq:euler_F}\\
	HF_{HH}+XF_{HX}&=0. \label{eq:euler_FH}
\end{align}
\end{subequations}

Let $h=\log(H/X)$. Using the definition \eqref{eq:sigma} and \eqref{eq:euler}, we obtain
\begin{align*}
	\frac{1}{\sigma_F}&=\frac{\partial}{\partial h}\log \frac{F_X(X\e^h,X)}{F_H(X\e^h,X)}=\frac{X\e^h F_{HX}}{F_X}-\frac{X\e^h F_{HH}}{F_H}\\
	&=\frac{HF_{HX}}{F_X}-\frac{HF_{HH}}{F_H}=\frac{HF_{HX}}{F_X}+\frac{XF_{HX}}{F_H}\\
	&=\frac{F_{HX}}{F_HF_X}(HF_H+XF_X)=\frac{FF_{HX}}{F_HF_X}. \qedhere
\end{align*}
\end{proof}

\subsection{Proof of Theorem \ref{thm:main}}

We prove Theorem \ref{thm:main} by establishing a series of lemmas.

\begin{lem}\label{lem:MRT}
Let $A>0$ and suppose that $\sigma_F(H,1)\ge \sigma$ for $H\ge A$. Let $\rho=1/\sigma$. If $A_H/A_X\ge A$, then
\begin{equation}
    \frac{F_X}{F_H}(A_H,A_X)\le \frac{F_X}{F_H}(A,1)A^{-\rho}(A_H/A_X)^\rho. \label{eq:MRT_ub}
\end{equation}
\end{lem}

\begin{proof}
By Assumption \ref{asmp:F}, $F$ is homogeneous of degree 1. Therefore $F_H,F_X$ are homogeneous of degree 0, and so is $\rho(H,X)$ in \eqref{eq:sigma}.

Let $B\coloneqq A_H/A_X\ge A$. Setting $H=\e^h$ and $X=1$ in \eqref{eq:sigma}, we obtain
\begin{equation*}
    \rho(\e^h,1)=\frac{\diff}{\diff h}\log \frac{F_X}{F_H}(\e^h,1).
\end{equation*}
Integrating both sides from $h=\log A$ to $h=\log B$ and applying the intermediate value theorem for integrals, there exists $h_1\in (\log A,\log B)$ such that
\begin{align}
    \rho(\e^{h_1},1)\log (B/A)&=\int_{\log A}^{\log B} \rho(\e^h,1)\diff h \notag \\
    &=\log \frac{F_X}{F_H}(B,1)-\log \frac{F_X}{F_H}(A,1). \label{eq:rho_h1}
\end{align}
Taking the exponential of both sides of \eqref{eq:rho_h1}, letting $M\coloneqq (F_X/F_H)(A,1)$, and using the homogeneity of $F_H,F_X$, we obtain
\begin{equation*}
    \frac{F_X}{F_H}(A_H,A_X)=\frac{F_X}{F_H}(B,1)=M (B/A)^{\rho(\e^{h_1},1)}.
\end{equation*}
Since $B\ge A$ and $\rho(\e^{h_1},1)\le \rho\coloneqq 1/\sigma$, it follows that
\begin{equation*}
    \frac{F_X}{F_H}(A_H,A_X)\le M (B/A)^\rho=MA^{-\rho}(A_H/A_X)^\rho,
\end{equation*}
which is \eqref{eq:MRT_ub}.
\end{proof}

\begin{lem}\label{lem:V_ub}
In equilibrium, the fundamental value of land is bounded above as
\begin{equation}
    V_t\le w_t\E_t\left[\sum_{s=1}^\infty \frac{r_{t+s}}{w_{t+s}}\right]. \label{eq:V_ub}
\end{equation}
\end{lem}

\begin{proof}
The stochastic discount factor between time $t$ and $t+1$ equals the marginal rate of substitution
\begin{align*}
    m_{t\to t+1}&\coloneqq \frac{\beta/c_{t+1}^o}{(1-\beta)/c_t^y}=\frac{\beta c_t^y}{(1-\beta)c_{t+1}^o}\\
    &=\frac{\beta w_t}{\beta w_{t+1}+r_{t+1}}\le \frac{w_t}{w_{t+1}},
\end{align*}
where the last line uses \eqref{eq:eq_obj} and $r_{t+1}\ge 0$. Then we can bound the stochastic discount factor between time $t$ and $t+s$ from above as
\begin{equation*}
    m_{t\to t+s}\coloneqq \prod_{j=0}^{s-1}m_{t+j\to t+j+1}\le \frac{w_t}{w_{t+s}}.
\end{equation*}
Therefore we can bound the fundamental value of land from above as
\begin{align*}
    V_t&\coloneqq \E_t\left[\sum_{s=1}^\infty m_{t\to t+s}r_{t+s}\right]\\
    &\le \E_t\left[\sum_{s=1}^\infty \frac{w_t}{w_{t+s}}r_{t+s}\right]=w_t\E_t\left[\sum_{s=1}^\infty \frac{r_{t+s}}{w_{t+s}}\right]. \qedhere
\end{align*}
\end{proof}

\begin{lem}\label{lem:VP}
We have $\lim_{t\to\infty}V_t/P_t=0$ almost surely.
\end{lem}

\begin{proof}
By \eqref{eq:P} and Lemma \ref{lem:V_ub}, we have
\begin{equation*}
    0\le \frac{V_t}{P_t}\le \frac{1}{\beta}\E_t\left[\sum_{s=1}^\infty \frac{r_{t+s}}{w_{t+s}}\right].
\end{equation*}
Therefore to show the claim, it suffices to show that $\E_t\left[\sum_{s=1}^\infty r_{t+s}/w_{t+s}\right]\to 0$ almost surely as $t\to\infty$.

By Assumption \ref{asmp:ES}, we can take a constant $A>0$ such that $\sigma(H,1)\ge \sigma>1$ for all $H\ge A$. Let $A_t\coloneqq A_{Ht}/A_{Xt}$ and $\rho=1/\sigma\in (0,1)$. Since the expectation of the infinite sum \eqref{eq:cond_A} is finite, the sum converges with probability 1 and hence we must have $A_t^{\rho-1}\to 0$ and $A_t\to \infty$ because $\rho<1$. In particular, there exists $T>0$ such that $A_t\ge A$ for $t\ge T$. For such $t$, by Lemma \ref{lem:MRT} we have
\begin{equation*}
    \frac{r_t}{w_t}=\frac{F_X(A_{Ht},A_{Xt})A_{Xt}}{F_H(A_{Ht},A_{Xt})A_{Ht}}\le \frac{F_X}{F_H}(A,1)A^{-\rho}A_t^{\rho-1}.
\end{equation*}
Therefore
\begin{equation*}
    \E_t\left[\sum_{s=1}^\infty \frac{r_{t+s}}{w_{t+s}}\right]\le \frac{F_X}{F_H}(A,1)A^{-\rho}\E_t\sum_{s=1}^\infty A_{t+s}^{\rho-1}.
\end{equation*}
Letting $t\to\infty$ and using condition \eqref{eq:cond_A}, we obtain $\E_t\left[\sum_{s=1}^\infty r_{t+s}/w_{t+s}\right]\to 0$ almost surely as $t\to\infty$.
\end{proof}

\begin{proof}[Proof of Theorem \ref{thm:main}]
The absence of arbitrage and the definition of the fundamental value imply
\begin{align*}
    P_t=\E_t[m_{t\to t+1}(P_{t+1}+r_{t+1})],\\
    V_t=\E_t[m_{t\to t+1}(V_{t+1}+r_{t+1})].
\end{align*}
Taking the difference, we obtain
\begin{equation*}
    P_t-V_t=\E_t[m_{t\to t+1}(P_{t+1}-V_{t+1})].
\end{equation*}
Iterating this equation and applying the law of iterated expectations, we obtain
\begin{equation*}
    P_t-V_t=\E_t[m_{t\to t+s}(P_{t+s}-V_{t+s})].
\end{equation*}
Lemma \ref{lem:VP} implies $V_{t+s}/P_{t+s}\to 0$ almost surely as $s\to\infty$ and hence $P_{t+s}>V_{t+s}$ for large enough $s$ with probability 1. Therefore $P_t>V_t$ for all $t$, and land is overvalued.
\end{proof}

\subsection{Proof of Proposition \ref{prop:ES}}

Let $F$ be as in \eqref{eq:FE}. Then
\begin{align*}
    F_H=&\alpha H^{\alpha-1}E^{1-\alpha}+(1-\alpha)H^\alpha E^{-\alpha}E_H\\
    =&H^\alpha E^{-\alpha}(\alpha H^{-1}E+(1-\alpha)E_H),\\
    F_X=&(1-\alpha)H^\alpha E^{-\alpha}E_X,\\
    F_{HX}=&\alpha(1-\alpha)H^{\alpha-1}E^{-\alpha}E_X\\
    &-\alpha(1-\alpha)H^\alpha E^{-\alpha-1}E_HE_X+(1-\alpha)H^\alpha E^{-\alpha}E_{HX}\\
    =&(1-\alpha)H^\alpha E^{-\alpha}(\alpha H^{-1}E_X-\alpha E^{-1}E_HE_X+E_{HX}).
\end{align*}
Applying Lemma \ref{lem:sigma}, we obtain
\begin{equation*}
    \sigma_F=\frac{F_HF_X}{FF_{HX}}=\frac{(\alpha H^{-1}E+(1-\alpha)E_H)E_X}{E(\alpha H^{-1}E_X-\alpha E^{-1}E_HE_X+E_{HX})}.
\end{equation*}
Therefore
\begin{align*}
    \sigma_F-1&=\frac{E_{HX}}{E_{HX}+\alpha H^{-1}XE^{-1}E_X^2}(\sigma_E-1),
\end{align*}
where the last line uses $E=HE_H+XE_X$. Applying Lemma \ref{lem:sigma} again, we obtain \eqref{eq:sigmaF} and $\sigma_F-1$ and $\sigma_E-1$ have the same sign. \hfill \qedsymbol

\section{Stylized facts}\label{sec:data}

\subsection{U.S. data}\label{subsec:data_US}

To compute the employment share of land-intensive sectors in Figure \ref{fig:US}, we combine two data sources. The first is \emph{Historical Statistics of the United States, Colonial Times to 1970},\footnote{\url{https://www.census.gov/library/publications/1975/compendia/hist_stats_colonial-1970.html
}} Series D167--181 Labor Force and Employment, by Industry: 1800--1960. We use decennial data on employment (Series D170--181) from 1800 to 1920. The second is \emph{Bureau of Economic Analysis GDP \& Personal Income},\footnote{\url{https://www.bea.gov/itable/national-gdp-and-personal-income}} Section 6 Income and Employment by Industry, Tables 6.5A--6.5D Full-Time Equivalent Employees by Industry. We use annual data on employment since 1929.

From this raw data, we define four series that we call ``total'', ``agriculture'', ``mining'', and ``construction''. Before 1929, ``total'' is the sum of series 170--181; ``agriculture'' is the sum of Agriculture (170) and Fishing (171); ``mining'' is Mining (172); ``construction'' is Construction (173). Since 1929, ``total'' is line 1 in NIPA Table 6.5, ``agriculture'' is line 4, ``mining'' is line 7, and ``construction'' is line 12 (since 1948) or 13 (before 1948). Whenever there are duplicate years in the tables, we use data from the more recent table. Finally, we define ``land-intensive'' by the sum of ``agriculture'', ``mining'', and ``construction''. After defining these series, we compute the employment shares of ``agriculture'' and ``land-intensive'' by dividing by ``total''.

Real per capita GDP since 1800 is from \emph{Maddison Project Database 2023}.\footnote{\url{https://www.rug.nl/ggdc/historicaldevelopment/maddison/releases/maddison-project-database-2023}} We convert the amounts to 2023 dollars using the consumer price index.\footnote{\url{https://fred.stlouisfed.org/series/CPIAUCSL}}

\subsection{Cross-country data}\label{subsec:data_country}

To create Figure \ref{fig:country}, we use the World Development Indicator from the World Bank.\footnote{\url{https://wdi.worldbank.org/table/}} We obtain population from Table 2.1. We obtain the employment share of agriculture from Table 2.3 by taking the average for male and female. We obtain the GDP and GDP share of agriculture from Table 4.2. Finally, we compute the per capita GDP by dividing GDP by population. Figure \ref{fig:country} plots these variables for countries without missing values (173 in total).

\printbibliography

\end{document}